\documentclass[twoside]{article}
\usepackage[preprint]{aistats2027}
\usepackage{comment}
\usepackage[round]{natbib}
\usepackage{xcolor}
\usepackage[
    colorlinks=true,
    citecolor=blue,
    linkcolor=blue,
    urlcolor=blue
]{hyperref}

\usepackage{amsmath,amssymb,amsthm,mathtools,bm}
\usepackage{float}
\usepackage{graphicx,booktabs,microtype,url}
\usepackage{tikz}
\usetikzlibrary{positioning}
\usetikzlibrary{arrows}
\graphicspath{{/figures/}}
\newtheorem{theorem}{Theorem}
\newtheorem{proposition}{Proposition}
\newtheorem{corollary}{Corollary}
\newtheorem{definition}{Definition}

\newcommand{\E}{\mathbb E}
\newcommand{\F}{\mathcal F}
\newcommand{\KL}{\operatorname{KL}}
\newcommand{\dd}{\,\mathrm d}

\begin{document}

%

%

\twocolumn[
\aistatstitle{Coherent Information Deletion: Bayesian, 
Generalized Bayesian and Possibilistic Unlearning}
\aistatsauthor{Hans Montcho \And H\r{a}vard Rue \And Maurizio Filippone}
\aistatsaddress{\\ KAUST, Saudi Arabia}
]

\begin{abstract}
Bayes’ theorem admits an information processing interpretation due to 
\citet{zellner1988optimal}: under the Shannon information criterion, the
posterior is the unique rule that processes prior and data information 
without information loss. We revisit these ideas 
from the perspective of information deletion. Given a posterior based on 
a complete dataset, what distribution should replace
it when a subset of the data is removed? For ordinary Bayesian inference, 
we show that the corresponding information deletion 
objective is uniquely minimized  by the leave-data-out posterior. We then
extend the framework beyond likelihood-based inference from Bayes
to the generalized Bayesian updating of \citet{bissiri2016general} based on 
loss functions.  Requiring deletion to be coherent under sequential 
removal of loss contributions yields a deletion rule that exactly 
recovers the generalized posterior based only on the retained information.
We further develop an axiomatic theory, beyond probabilistic inference,
for possibilistic deletion and show that the same information-deletion 
perspective leads to a principled rule for exact possibilistic unlearning.
Finally, restricting these exact deletion principles to tractable families 
leads naturally to variational Bayesian, generalized Bayesian, and maxitive
possibilistic unlearning.

\end{abstract}

\section{Introduction}
\label{sec:introduction}

Statistical learning can be viewed as the processing of information 
into a state of belief. In ordinary Bayes, where both prior knowledge and
data information are represented in terms of probability, 
\citet{zellner1988optimal} showed that Bayes' theorem emerges, as the most efficient
information processing rule, from an optimal Shannon information processing 
principle. Generalized Bayesian inference relaxes this probabilistic
assumption, by allowing the data information to be encoded through a loss
function; and under a sequential coherence requirement, 
\citet{bissiri2016general} obtained  a valid generalized posterior belief. More
recently, related information processing principles have been 
developed beyond probabilistic belief states. In possibilistic inference 
\citep{zadeh1999fuzzy,martin2026possibilistic}, uncertainty is represented
by possibility functions and normalization is maxitive rather than additive, with 
recent work deriving principled possibilistic updating rules and 
corresponding variational formulations \citep{houssineau2026bayes, 
singh2025maxitive}. Across these settings, the question has been on incorporating
information, and this work studies the inverse problem: how should information
that has already been  added be optimally removed?

This inverse problem naturally appears in
cross-validation \citep{vehtari2012survey} and 
machine unlearning \citep{bourtoule2021machine}. In both cases, one starts
from a fitted belief based on  complete information and seeks the belief
that would have been obtained had a specified contribution never
been incorporated. Existing Bayesian unlearning methods are often 
introduced as approximations to the retrained posterior 
\citep{nguyen2020variational,rawat2022challenges}. Here, we first 
characterize the exact deletion operation from an information 
and axiomatic principle, and only afterward introduce approximation.

For ordinary Bayes, we formulate information deletion using Zellner's 
Shannon information accounting and   
investigate how to transform a posterior distribution into a
post-deletion distribution  when a subset of  data is removed.
We show that the resulting objective is a Kullback-Leibler (KL) divergence 
whose unique minimizer is the leave-data-out posterior. 
This result gives a learning-unlearning duality of 
Bayes' theorem in a precise sense: it is optimal both for adding
information and for removing it.

A natural  question is whether this information deletion view of Bayes'
theorem extends beyond likelihood-based updating. The generalized posterior 
belief of \citet{bissiri2016general}  is defined as the solution of a decision
problem combining expected data loss with deviation from the prior 
belief. A key requirement in their construction is
coherent updating, by which, sequentially processing
two pieces of information  must yield the same result as processing
them jointly. This requirement characterizes the KL divergence, up to a scale, 
as the appropriate penalty for departing from the prior belief.
We also impose the analogous requirement of coherent information deletion, 
implying that sequentially deleting two pieces of information 
must be equivalent to  their joint deletion.  
Within the same class of divergence-based penalties, 
coherent deletion  also forces the discrepancy from the current posterior 
belief to be proportional to a KL divergence. 
The resulting decision problem yields a  
generalized Bayesian unlearning rule and, when applied to a generalized 
posterior, recovers exactly the generalized Bayesian update based only 
on the retained information.

The probabilistic theory also suggests an extension beyond probability measures.
\citet{houssineau2026bayes} recently derived a  possibilistic learning rule from 
an information processing and axiomatic viewpoint for arbitrary losses; and 
possibilistic unlearning was explicitly left as an open inverse 
problem. We take the axiomatic route, and  require analogous possibilistic
coherent deletion axioms, showing that the same structural principles
used to characterize possibilistic learning can also determine how
previously incorporated information should be removed;
thereby yielding a principled possibilistic unlearning rule.

Finally, the learning rules also provides optimization characterizations that
form the basis of variational inference. In the probabilistic setting, 
Bayesian updating can be viewed as the solution of an 
optimization problem over probability distributions, with variational
inference obtained by restricting this problem to a 
tractable family \citep{blei2017variational, knoblauch2022optimization}.
\citet{singh2025maxitive} established the  optimization characterization 
for possibilistic learning, leading to a maxitive form of variational inference.
Our deletion principles inherit the same structure, and once the exact 
post-deletion belief is characterized, restricting the corresponding 
optimization problem to a tractable family yields Bayesian, 
generalized Bayesian, and possibilistic variational unlearning method.

The  contributions of this paper are twofold.

1. We develop a formal theory of coherent information deletion 
across probabilistic and possibilistic inference. 
The framework recovers leave-data-out Bayes in the
ordinary Bayesian case, extends deletion to generalized Bayesian updating
through coherence principles, and provides an axiomatic foundation 
for possibilistic unlearning.

2. We show how these exact deletion principles lead naturally
to tractable approximation. Restricting the corresponding optimization
problems to suitable families yields variational unlearning formulations
for Bayesian, generalized Bayesian, and possibilistic inference.

Technical proofs are given  in the Supplementary Material.

\begin{figure*}[t]
\centering

\begin{minipage}[t]{0.48\textwidth}
\centering
\begin{tikzpicture}[
  box/.style={
    draw,
    rectangle,
    text width=1.7cm,
    align=center,
    minimum height=1.25cm,
    font=\scriptsize
  },
  node distance=0.6cm
]

\node[box] (A) {
   \textbf{Inputs}\\[1mm]
   $\pi(\theta\mid I)$\\
   $\pi(y\mid\theta)$
};

\node[box, right=of A] (B) {
   \textbf{Information}\\
   \textbf{Processing}\\
   \textbf{Rule}
};

\node[box, right=of B] (C) {
   \textbf{Outputs}\\[1mm]
   $q(\theta\mid y)$\\
   $\pi(y\mid I)$
};

\draw[->,thick] (A) -- (B);
\draw[->,thick] (B) -- (C);

\end{tikzpicture}

\vspace{1mm}
\textbf{(a) Information addition}
\end{minipage}
\hfill
\begin{minipage}[t]{0.48\textwidth}
\centering
\begin{tikzpicture}[
  box/.style={
    draw,
    rectangle,
    text width=1.7cm,
    align=center,
    minimum height=1.25cm,
    font=\scriptsize
  },
  node distance=0.6cm
]

\node[box] (A) {
   \textbf{Inputs}\\[1mm]
   $\pi(\theta\mid y)$\\
   $\pi(y\mid I)$\\
   $\pi(y_g\mid\theta,y_{-g})$
};

\node[box, right=of A] (B) {
   \textbf{Information}\\
   \textbf{Deletion}\\
   \textbf{Rule}
};

\node[box, right=of B] (C) {
   \textbf{Outputs}\\[1mm]
   $q(\theta\mid y_{-g})$\\
   $\pi(y_{-g}\mid I)$
};

\draw[->,thick] (A) -- (B);
\draw[->,thick] (B) -- (C);

\end{tikzpicture}

\vspace{1mm}
\textbf{(b) Information deletion}
\end{minipage}

\caption{
Information addition and deletion. 
Panel (a) summarizes Zellner's information processing rule;
panel (b) depicts the corresponding information deletion problem.
}
\label{fig:add-delete}
\end{figure*}

\section{Optimal Information deletion and Bayes' Theorem}
\label{sec:ordinary-bayes}
\subsection{Bayesian information processing}

Let $\theta\in\Theta$ be an unknown quantity,  $I$ be prior 
information, and let $\pi(\theta\mid I)$ and $\pi(y\mid\theta)$
denote the prior distribution and likelihood.  Following \citet{zellner1988optimal},
prior and data  represent the input information to be transformed into
some output information represented by the postdata distribution 
$q(\theta\mid y)$ and the probability distribution for the data 
$\pi(y \mid I) = \int \pi(y \mid \theta)\pi(\theta \mid I) \dd\theta$.
An information processing rule  (IPR), see  Figure \ref{fig:add-delete},
is a procedure that transforms
input into  output information, and the optimal IPR is  one that does 
not destroy or add  nonexistent information. 

Using Shannon information, Zellner defines information loss as the 
difference between output and input information, and the corresponding 
criterion is
\begin{align*}
\Delta_{\mathrm{add}}(q)
&= \text{Output Information}
   - \text{Input Information} \\
&= \E_{q}
   \left[
   \log q(\theta\mid y)
   + \log \pi(y\mid I)
   \right] \\
&\quad -
   \E_{q}
   \left[
   \log \pi(\theta\mid I)
   + \log \pi(y\mid\theta)
   \right] \\
&= \E_{q}
   \left[
   \log
   \frac{q(\theta\mid y)}{\pi(\theta\mid I)}
   \frac{\pi(y\mid I)}{\pi(y\mid\theta)}
   \right].
\end{align*}

Its unique minimizer, among all candidate postdata distributions on $\Theta$, is
the Bayesian posterior 
$$q^*(\theta \mid y) = \pi(\theta\mid y) =\frac{\pi(\theta\mid I)
\pi(y \mid \theta)}{\pi(y \mid I)}$$ 
making Bayes' theorem the optimal IPR.

An important consequence is that Bayes' theorem can be viewed as 
the solution to an infinite-dimensional optimization problem.
This immediately suggests  approximate solutions obtained by 
restricting the optimization to a tractable family $\mathcal Q$
\begin{equation}
\tilde q(\theta\mid y) =
\underset{q\in\mathcal Q}{\arg\min}\;
\Delta_{\mathrm{add}}(q) \nonumber
\end{equation}
providing an information processing interpretation of 
variational inference \citep{blei2017variational}.

The rest of this paper is concerned with the inverse operation, namely, 
on the existence of an analogous principle to update the posterior density
when part of the information already incorporated needs to be removed.

\subsection{Information deletion}

Information deletion arises whenever an inferential procedure must 
remove the contribution of data that have already been incorporated
into a fitted belief. Two important examples are cross-validation and 
machine unlearning. In cross-validation \citep{vehtari2012survey},
one seeks the posterior that would have been obtained without a 
held-out subset in order to assess predictive performance, 
often using approximations to avoid repeated refitting. 
In machine unlearning \citep{bourtoule2021machine}, 
the objective is instead to remove specified data from an already fitted 
and possibly deployed model without retraining from scratch.
Despite their different purposes, both problems require the same 
underlying operation: transforming a full-data posterior into 
the belief that corresponds to deleting a specified information contribution. 
Our goal is to formalize this operation from first principles, 
identify the exact post-deletion belief, and only then 
consider tractable approximations to it.

\subsection{Information deletion loss}

The inputs to the information deletion rule (IDR), depicted in Figure 
\ref{fig:add-delete}, are the outputs from the 
IPR, i.e. the posterior $\pi(\theta \mid y)$ and the  marginal density
$\pi(y|I)$ under  prior information $I$ .  
We partition the data as $y=(y_g, y_{-g})$, where $y_g$ represents  
the subset  whose information content needs to be removed, and we 
denote its likelihood contribution by $ \ell_g(\theta) =\pi(y_g\mid\theta, y_{-g})$.
The outputs of the IDR are $q(\theta\mid y)$ a candidate post-deletion  probability 
density, and the marginal density of the remaining data  
$\pi(y_{-g}\mid I)= \int_\Theta \pi(y_{-g}\mid\theta)\pi(\theta\mid I)\dd \theta$.

By analogy with the information accounting used for Bayesian updating, 
we define the  information deletion loss 
\begin{eqnarray}
    \Delta_g(q) =& \text{Output Information -  Input Information}  \nonumber\\
    =& \mathbb{E}_{q}\bigg[\log q(\theta\mid y_{-g}) + \log\pi(y_{-g}\mid I)\bigg] 
    \nonumber \\
    -& \mathbb{E}_{q}\bigg[\log  \pi(\theta \mid y) + \log \pi(y\mid I) -
    \log \ell_g(\theta) \bigg] \label{eq-removal_info_yg}   
\end{eqnarray}

and  we note the \emph{negative sign(-)} in front of $\ell_g(\theta)$ 
in Equation \eqref{eq-removal_info_yg}  due to the removal of its 
information content. Under this loss, the optimal IDR will be obtained as 
the unique minimizer of $\Delta_g(q)$.

\subsection{The deletion distribution}

Define, on the support where $\ell_g(\theta) > 0$, 
\begin{equation}
 r_g(\theta)=\frac{\pi(\theta\mid y)}{\ell_g(\theta)},
\qquad
C_g^B=\int_\Theta r_g(\theta)\dd\theta,
\label{eq:rg}
\end{equation}
and assume $0 < C_g^B < \infty$.
The following observation is the central technical step.

\begin{proposition}[Information deletion loss is a KL divergence]
\label{prop:ordinary-kl}
Under the preceding assumptions, for every candidate density $q$,
with $r_g(\theta)>0$ wherever $q>0$
\begin{equation}
\Delta_g(q)=\KL\!\left(q\,\middle\|\,\frac{r_g(\theta)}{C_g^B}\right).
\label{eq:kl-form}
\end{equation}
\end{proposition}
Proposition \ref{prop:ordinary-kl} frames the deletion loss as a divergence
with respect to the posterior density given the complete data.  
Its proof is given in the supplementary material.  

\subsection{Optimal IDR and leave-data-out Bayes}

We can now state the main result.

\begin{theorem}[Optimal information deletion]
Under the assumptions of Proposition~\ref{prop:ordinary-kl}, the 
unique minimizer of the information-deletion loss $\Delta_g(q)$ 
over all probability densities on $\Theta$ is

\begin{equation}
 q_g^*(\theta\mid y)
=\frac{\pi(\theta\mid y)/\ell_g(\theta)}
{\int_\Theta \pi(\theta\mid y)/\ell_g(\theta)\dd\theta},
\end{equation}
and the minimum information loss is zero.
\end{theorem}

\begin{proof}
By Proposition~\ref{prop:ordinary-kl}, $\Delta_g(q)=\KL(q\|q_g^*)$ in 
the extended-value sense. Non-negativity of the KL divergence gives 
$\Delta_g(q)\geq0$, with equality if and only if $q=q_g^*$ almost everywhere. 
Hence $q_g^*$ is the unique global minimizer and the minimum value is zero.
\end{proof}

\begin{corollary}[Equivalence with leave-data-out Bayes]
\label{cor:eq-leave-bayes}
Under the same assumptions,
\begin{equation}
q_g^*(\theta\mid y)=\pi(\theta\mid y_{-g})
\label{eq:main-equivalence}
\end{equation}
almost everywhere. 
\end{corollary}
Thus the optimal information deletion rule is exactly the leave-data-out posterior.
 
If the data are conditionally independent given $\theta$, then the 
result takes the familiar likelihood division form
\begin{equation}
\label{eq-ratio_is}
\pi(\theta\mid y_{-g}) 
\propto
\frac{\pi(\theta\mid y)}{\pi(y_g\mid\theta)}, \nonumber
\end{equation}
which is often the importance sampling proposal used in 
approximate leave-data-out cross validation. Conditional independence
is therefore not part of the information deletion principle; 
it is only a simplification of the data contribution. 

\section{Coherent generalized information deletion}
\label{sec:coherent-generalized}

The preceding analysis starts from Zellner's Shannon information accounting and
shows that ordinary Bayesian deletion recovers the leave-data-out posterior.
We now allow  for the data information to be incorporated
through a loss function rather than a likelihood, and we search for a 
deletion rule  justified by
a coherence principle. The construction follows the decision-theoretic route of 
\citet{bissiri2016general}.

\subsection{Generalized Bayesian update}

Let $P_0$ denote the prior belief on $\Theta$,  $\ell(\theta,y)$ be a loss
representing the information supplied by the data $y=(y_g, y_{-g})$, 
and let $P$ denote the updated belief.
In the framework of \citet{bissiri2016general}, the updated belief is defined
as the minimizer of a decision loss with two components: the expected data
loss and a discrepancy from the prior belief. To obtain the discrepancy, they
require a coherent updating procedure, in the sense that for an 
updating operator $\psi$
\begin{align*}
\psi\!\left\{
\ell(\theta,y_{-g}),
\psi\bigl(\ell(\theta,y_g),P_0\bigr)
\right\}&=
\psi\!\left\{
\ell(\theta,y),P_0
\right\},
\end{align*}
meaning that,  processing $y_g$ and then $y_{-g}$ gives the same posterior 
belief  as processing their combined loss in a single step. Then within the
class of differentiable divergences,  they showed that the discrepancy must
be a positively scaled KL divergence, and the generalized Bayesian update becomes
\begin{equation}
P
=
\arg\min_Q
\left\{
w\,\E_Q[\ell(\theta,y)]
+
\KL(Q\|P_0)
\right\}
\end{equation}
where $w >0$ is a learning rate controlling the relative weight of the data
loss to the prior. In addition, If $P_0$ has density $p_0$, the resulting
posterior belief has  density
\begin{equation}
p(\theta)
=
\frac{
p_0(\theta)\exp\{-w\ell(\theta,y)\}
}{
\int p_0(\theta)\exp\{-w\ell(\theta,y)\}\,\dd\theta
}. \label{eq:gen-bayes}
\end{equation}
We now consider the inverse problem: starting from the current belief $P$,
how should the contribution of a loss  already incorporated into it be removed?

\subsection{From coherent updating to coherent deletion}

For deletion we impose the corresponding inverse requirement.
\begin{definition}[Sequentially coherent deletion]
Let $\mathcal D(\ell,P)$ denote the belief distribution obtained
by removing information represented
by a loss $\ell$ from the current belief $P$. A deletion 
rule is said to be \emph{sequentially coherent} if, for any two removable 
loss contributions $\ell_1$ and $\ell_2$,
\begin{equation}
\mathcal D\left(\ell_2, \mathcal D(\ell_1,P)\right)
=
\mathcal D(\ell_1+\ell_2, P),
\label{eq:deletion-coherence}
\end{equation}
whenever the involved deletion operations are well-defined.
\end{definition}
In words,  removing two pieces of information successively must be equivalent to 
removing their combined information at once. Because loss contributions
are additive, \eqref{eq:deletion-coherence} also implies order invariance
whenever both orders are admissible.

\subsection{A decision loss for deletion}

As in \citet{bissiri2016general}, suppose the post-deletion distribution $Q$ 
is chosen by minimizing a decision loss with two components,
\begin{equation}
\mathcal L_{-}(Q;P,\ell)
=
h_1^{-}(Q,\ell )+h_2(Q,P).
\label{eq:generic-deletion-loss}
\end{equation}
The term $h_1^{-}$ represents the information to be removed, 
whereas $h_2$ penalizes departure from the current belief $P$.
For generalized updating, the data $y_g$ contribution is the expected
loss $\E_Q[\ell(\theta,y_g)]$. Once this information has already been incorporated, 
deleting the same contribution must reverse its directional effect.
We therefore take
\begin{equation}
h_1^{-}(Q,\ell)
=-\int \ell(\theta,y_g)\,Q(\dd\theta)
=-\E_Q[\ell(\theta,y_g)].
\label{eq:h1-deletion}
\end{equation}
This choice preserves the additive representation of information:
$$
h_1^{-}(Q, \ell_1 + \ell_2)
=h_1^{-}(Q,\ell_1)+h_1^{-}(Q,\ell_2),
$$
and is the  sign reversal of the expected-loss term used when 
information is added. 

The remaining question is the form of $h_2$, 
and  we do not assume it to be given by a KL divergence. Instead,
we restrict $h_2$ to the class of differentiable $g$-divergences,
\begin{equation}
h_2(Q,P)
=D_g(Q,P)
=
\int g\!\left(\frac{\dd Q}{\dd P}\right)\dd P,
\label{eq:g-divergence}
\end{equation}
for $Q\ll P$, where $g$ is differentiable and convex with $g(1)=0$, and we 
derive the discrepancy  for which the deletion rule is 
sequentially coherent.

\begin{theorem}[Coherent deletion characterizes the deletion discrepancy]
\label{thm:coherent-deletion-kl}
Define the post-deletion belief after removing the loss $\ell(\theta, .)$ by 
\begin{equation}
\mathcal D(\ell,P)
=
\arg\min_Q
\left\{
-\E_Q[\ell(\theta, .)] + D_g(Q,P)
\right\}.
\label{eq:deletion-g-objective}
\end{equation}
Assume that the relevant
minimizers exist uniquely and that the deletion operator satisfies
sequential coherence \eqref{eq:deletion-coherence} 
for every admissible $P$ and $\ell$. Then, there exists a constant $c>0$
such that 
\begin{equation}
D_g(Q,P) =  c\KL(Q\|P).
\label{eq:coherence-kl-result}
\end{equation}
\end{theorem}

Essentially, coherence determines the discrepancy only up to scale, and for a loss 
$\ell_g = \ell(\theta, y_g)$ to be removed, 
substituting \eqref{eq:coherence-kl-result} into \eqref{eq:deletion-g-objective}
gives
$$
\mathcal D_{v}(\ell_g,P) = \arg\min_Q\left\{-v\E_Q[\ell_g] + \KL(Q\|P)\right\}, 
\quad v = \frac{1}{c}.
$$
In the next section, we then solve the coherent deletion problem above.

\subsection{Solution of the coherent deletion problem}

Let $P$ have density $p$ and assume
\begin{equation}
C_g(v;P)
=
\int_{\Theta}
p(\theta)\exp\{v\ell_g(\theta)\}\dd\theta
<\infty .
\label{eq:deletion-normalizer}
\end{equation}

\begin{proposition}[Coherent generalized deletion]
\label{prop:generalized-deletion-solution}
The unique minimizer of
\begin{equation}
\mathcal D_{v}(\ell_g,P)
=
\arg\min_Q
\left\{
-v\E_Q[\ell_g]
+
\KL(Q\|P)
\right\}
\end{equation}
has density
\begin{equation}
q_g^*(\theta)
=
\frac{
p(\theta)\exp\{v\ell_g(\theta)\}
}{
C_g(v;P)
}.
\label{eq:generalized-deletion-solution}
\end{equation}
\end{proposition}

\begin{proof}
For any admissible density $q$,
\begin{align}
-v\E_Q[\ell_g]+\KL(Q\|P)
&=
\int q(\theta)
\log
\frac{q(\theta)}
{p(\theta)e^{v\ell_g(\theta)}}
\dd\theta \nonumber\\
&=
\KL(Q\|Q_g^*)
-
\log C_g(v;P).
\end{align}
Since the second term does not depend on $Q$, the result follows from
the non-negativity of the KL divergence.
\end{proof}

Thus coherent deletion exponentially reweights the current belief in the
opposite direction to generalized Bayesian learning.

The deletion rate $v$
remains undetermined by coherence alone, and to relate it to the learning
rate $w$ of the corresponding generalized Bayesian update we require 
inverse consistency.
If
\[
\mathcal \psi_{w}(\ell,P_0)
\propto
P_0 e^{-w\ell},
\qquad
\mathcal D_{v}(\ell,P)
\propto
P e^{v\ell},
\]
we require deletion of the same information to  invert learning
\[
\mathcal D_{v}\left(\ell,\;
\mathcal \psi_{w}(\ell,P_0)
\right)
=
P_0,
\]
meaning that in the extreme case, deleting the total loss $\ell(\theta, y)$
from $P$ should recover the prior belief $P_0$.
For arbitrary nonconstant $\ell$, this implies
\begin{align}
v=w, \text{ and } q_g^*(\theta)
=\frac{p(\theta)\exp\{w\ell_g(\theta)\}
}{
\int_{\Theta}
p(\theta)\exp\{w\ell_g(\theta)\}\dd\theta
}.\label{eq:optimal-gen-deletion}
\end{align}
\subsection{Equivalence with generalized Bayesian updating on the retained data}

\begin{corollary}
\label{cor:generalized-unlearning}
If the current belief is the generalized posterior \eqref{eq:gen-bayes}
formed with learning rate $w$, then coherent deletion of $\ell_g$ using the 
same $w$ exactly recovers the generalized posterior based 
only on the retained information 
\[
q_g^*(\theta) = \psi_{w}(\ell_{-g}(\theta), p_0).
\]
\end{corollary}
Ordinary Bayesian deletion is recovered with $w=1$ and
\[
\ell_g(\theta)=-\log \pi(y_g\mid\theta,y_{-g}),
\]
in which case $e^{w\ell_g(\theta)}$ divides the full posterior by 
the conditional likelihood contribution.

In summary, these  results show that coherent information deletion can be derived 
within probabilistic belief systems, both for ordinary and generalized Bayes.
In the next section, we extend the  deletion principle  to the possibilistic 
inferential paradigm.

\section{Axiomatic possibilistic information deletion}
\label{sec:poss-deletion}

We now consider possibility functions to represent beliefs about $\theta$, 
for which normalization is maxitive rather than additive. We will define the
possibilistic analogous of Bayes' theorem and also present its variational 
characterization.
Let
\[
\mathcal F(\Theta)
=
\left\{f:\Theta\to[0,1]:\sup_{\theta\in\Theta}f(\theta)=1\right\}.
\]
A function $f\in\F(\Theta)$ represents epistemic compatibility, 
and the possibility of a set $A$ is $\sup_{\theta\in A}f(\theta)$.

\subsection{Possibilistic learning}

For a prior possibility function $f_0$ and loss $\ell(\theta,y)$,
\citet{houssineau2026bayes} derives the possibilistic update
as
\begin{equation}
\begin{aligned}
\rho(\theta) = \psi_w[\ell,f_0](\theta)
&=
\frac{f_0(\theta)e^{-w\ell(\theta,y)}}{Z_{\max}(w;f_0,\ell)},\\
Z_{\max}(w;f_0,\ell)
&=
\sup_{\theta\in\Theta}f_0(\theta)e^{-w\ell(\theta,y)}.
\end{aligned}
\label{eq:poss-update}
\end{equation}
where $w>0$ is a learning rate parameter. Their axiomatic derivation uses
sequential coherence, locality under restriction, monotonicity in loss,
invariance to vacuous losses,  shift invariance; and our approach to derive
a principled possibilistic deletion update  also takes this route.

\subsection{Maxitive Donkser-Varadhan formulation of possibilistic learning}

\citet{singh2025maxitive} showed that a  possibilistic update of the form 
\eqref{eq:poss-update} also admits a maxitive variational characterization.
Define the lower and upper consistency-bound objectives
\begin{align}
\underline{\mathrm{CBO}}(f)
&=
\inf_{\theta\in\Theta}
\left\{
-w\ell(\theta)
-\log\frac{f(\theta)}{f_0 (\theta)}
\right\}, \\
\overline{\mathrm{CBO}}(f)
&=
\sup_{\theta\in\Theta}
\left\{
-w\ell(\theta)
-\log\frac{f(\theta)}{f_0 (\theta)}
\right\}.
\end{align}

A simplified form of their main result states  that
\begin{equation}
\begin{aligned}
\log Z_{\max}
&=
\sup_{f\in\mathcal F(\Theta)} \underline{\mathrm{CBO}}(f)\\
&=
\inf_{f\in\mathcal F(\Theta)} \overline{\mathrm{CBO}}(f).
\end{aligned}
\label{eq:maxitive-dv}
\end{equation}
Also, let the max-relative entropy be
\begin{equation}
D_{\max}(f\|h)
=
\sup_{\theta\in\Theta}
\log\frac{f(\theta)}{h(\theta)},
\end{equation}
then
\begin{align}
\log Z_{\max}
&=
\underline{\mathrm{CBO}}(f)
+
D_{\max}(f\|\rho), \\
\log Z_{\max}
&=
\overline{\mathrm{CBO}}(f)
-
D_{\max}(\rho \|f).
\end{align}

Hence maximizing the lower CBO is equivalent to minimizing
$D_{\max}(f\|\rho)$, while minimizing the upper CBO is equivalent to
minimizing $D_{\max}(\rho\|f)$. Restricting $f$ to a tractable family
$\tilde{\mathcal{F}}$ therefore yields two maxitive variational projections,
analogous to variational inference in the probabilistic setting. Later on,
we will also derive, to the best of our knowledge, 
the first  maxitive variational  characterization for 
possibilistic unlearning.

\subsection{Deletion axioms}

Let $\mathcal D(\ell,\rho)\in\F(\Theta)$ denote the possibility function obtained
after removing information represented by a loss $\ell$ from a current 
possibility function $\rho$.

For $A\subseteq\Theta$ satisfying $\sup_{\theta\in A}\rho(\theta)>0$, define
the restricted and renormalized possibility function
\[
\rho_A(\theta)
=
\frac{\mathbf 1_A(\theta)\rho(\theta)}
{\sup_{\theta\in A}\rho(\theta)}.
\]

\begin{definition}[Coherent possibilistic deletion]
\label{def:poss-axioms}
A possibilistic deletion operator $\mathcal D$ is said to satisfy the 
deletion axioms if the following hold whenever the relevant quantities 
are well-defined.
\begin{enumerate}
\item \emph{Sequential coherence:}
\[
\mathcal D(\ell_2,\;\mathcal D(\ell_1,\rho))
=
\mathcal D(\ell_1+\ell_2,\rho).
\]
\item \emph{Restriction compatibility:} for every $A$ with positive 
current possibility,
\[
\frac{\mathcal D(\ell,\rho)(\theta)}
{\sup_{\theta\in A}\mathcal D(\ell,\rho)(\theta)}
=
\mathcal D(\ell,\rho_A)(\theta),
\qquad \theta\in A.
\]
\item \emph{Relative monotonicity under deletion:} 
on every two-point restriction $\{a,b\}$ with positive current possibilities, if
\[
\ell(a)-\ell(b)>\ell'(a)-\ell'(b),
\]
then
\[
\frac{\mathcal D (\ell,\rho_{\{a,b\}})(a)}
{\mathcal D(\ell,\rho_{\{a,b\}}](b)}
>
\frac{\mathcal  D(\ell',\rho_{\{a,b\}})(a)}
{\mathcal D(\ell',\rho_{\{a,b\}})(b)}.
\]
Thus a larger loss contribution that is removed from $a$ relative to $b$ 
increases the post-deletion possibility odds of $a$ relative to $b$.
\item \emph{Vacuous-loss invariance:} if $\ell$ is constant, then 
$\mathcal D(\ell,\rho)=\rho$.
\item \emph{Shift invariance:} for every constant $c\in\mathbb R$,
\[
\mathcal D(\ell+c,\rho)=\mathcal D(\ell,\rho).
\]
\end{enumerate}
\end{definition}

The first, second, fourth and fifth conditions are direct deletion 
counterparts of the possibilistic coherence axioms of
\citet{houssineau2026bayes}. The third is stated at the pairwise-odds level
used by the characterization proof and reverses the direction of the 
learning monotonicity: removing a larger loss must increase, 
rather than decrease, relative possibility.

\begin{theorem}[Axiomatic possibilistic deletion]
\label{thm:poss-characterization}
Under the axioms in Definition~\ref{def:poss-axioms}, there exists a
constant $v>0$ such that, for every current possibility function $\rho$ 
and  loss $\ell_g$ satisfying
\[
0<C_{\max}(v;\rho,\ell)
:=
\sup_{\theta\in\Theta}\rho(\theta)e^{v\ell(\theta)}
<\infty,
\]
the possibilistic deletion rule is
\begin{equation}
\mathcal D_v(\ell,\rho)(\theta)
=
\frac{\rho(\theta)e^{v\ell(\theta)}}
{C_{\max}(v;\rho,\ell)}.
\end{equation}
\end{theorem}

The proof is given in Supplementary and  the main steps parallel the 
odds-based characterization of possibilistic learning.
As in the generalized Bayesian case, coherence determines the functional form 
but not the information strength parameter $v$, which is obtained using inverse
consistency. Therefore, for every  admissible $f_0$ and $\ell$,
$$\mathcal D_v(\ell,\; \psi_w(\ell,f_0)) = f_0
\quad
\text{if and only if } v = w.
$$

In summary, for a current possibilistic belief $\rho$, with learning rate $w$,
and loss $\ell_g(\theta) = \ell(\theta, y_g)$ to be  deleted, the  possibilistic
deletion rule is given by
\begin{equation}
\rho_{g}^* (\theta)
=
\frac{\rho(\theta)e^{w\ell_g(\theta)}}
{\sup_{\theta\in\Theta}\rho(\theta)e^{w\ell_g(\theta)}}.
\label{eq:poss-deletion-rule}
\end{equation}

\subsection{Equivalence with possibilistic  updating on the retained data}

\begin{corollary}
\label{cor:possibilistic-unlearning}
If the current possibility  $\rho$  given in  \eqref{eq:poss-update}
formed with learning rate $w$, then coherent deletion of $\ell_g$ using the 
same $w$ exactly recovers the possibilty function based 
only on the retained information 
\[
\rho_{g}^* (\theta) = \frac{
f_0(\theta)e^{-w\ell_{-g}(\theta)}
}{
\sup_{\theta\in\Theta}
f_0(\theta)e^{-w\ell_{-g}(\theta)}
}.\label{eq:}
\]
\end{corollary}

Thus possibilistic unlearning reconstructs exactly the possibility 
function that would have been obtained by processing only the retained loss.
So far, these results have characterized exact information deletion in Bayesian, 
generalized Bayesian, and possibilistic settings. Respectively, the post-deletion 
belief is obtained from an information processing, a decision theoretic and 
an axiomatic perspective. In all cases, the post-deletion belief is given as
solution to an optimization problem defined over the  space of admissible 
beliefs. Therefore, when the unrestricted problem is intractable, 
we can restrict the optimization to a tractable family and obtain 
corresponding variational unlearning approximations.

\section{ Variational Unlearning}
\label{sec:var-unlearn}
\subsection{Probabilistic variational unlearning}

For ordinary Bayes, recall that
\[
C_g^{\mathrm B}
=
\int_\Theta
\frac{\pi(\theta\mid y)}
{\pi(y_g\mid\theta,y_{-g})}\dd\theta,
\]
and for any admissible density $q$,
\begin{equation}
\begin{aligned}
\KL\!\left(q\middle\|\pi(\theta\mid y_{-g})\right)
={}&
\KL\!\left(q\middle\|\pi(\theta\mid y)\right)\\
&+\E_q\!\left[\log \pi(y_g\mid\theta,y_{-g})\right]
+\log C_g^{\mathrm B}.
\end{aligned}
\label{eq:ordinary-vb-identity}
\end{equation}
Hence, for a variational family $\tilde{\mathcal{Q}}$,
\begin{align}
\widetilde q_g
&=
\arg\min_{q\in \tilde{\mathcal{Q}}}
\KL\!\left(q\middle\|\pi(\theta\mid y_{-g})\right)
\label{eq:ordinary-vb-projection}\\
&=
\arg\min_{q\in \tilde{\mathcal{Q}}}
\left\{
\E_q\!\left[\log \pi(y_g\mid\theta,y_{-g})\right] +
\KL\!\left(q\middle\|\pi(\theta\mid y)\right)
\right\},
\label{eq:ordinary-vb-current}
\end{align}
because $\log C_g^{\mathrm B}$ does not depend on $q$.

Equation \eqref{eq:ordinary-vb-projection} gives a particularly simple
interpretation of variational Bayesian unlearning: 
it is the forward-KL projection of the exact retrained 
posterior onto the chosen variational family.
This distinction also clarifies the relationship 
with existing variational unlearning methods such as 
\citet{nguyen2020variational}. 
Those methods can be interpreted as approximations to the same 
target posterior, while the present formulation identifies 
the unrestricted target and the corresponding variational 
objective directly from an information conservation principle. 

In addition, Equation \eqref{eq:ordinary-vb-current} provides a 
practical interpretation of the
information deletion procedure. The minimization of the first term is 
equivalent to forgetting the information content of
$\pi(y_g\mid\theta,y_{-g})$, while the second term acts as a regularization,
preventing 
$\tilde{q}(\theta|y_{-g})$ from departing too much from the posterior
$\pi(\theta|y)$. Additionally, the regularization term  illustrates the
importance of an accurate posterior distribution, that is,
improved accuracy in  $\pi(\theta|y)$ should translate into 
better  approximations $\tilde{q}(\theta|y_{-g})$.

For generalized Bayes, Proposition~\ref{prop:generalized-deletion-solution},
$Z_{y}(w)= \int p_0(\theta)\exp\{-w\ell(\theta,y)\}\,\dd\theta$ 
and $Z_{-g}(w) =\int p_0(\theta)\exp\{-w\ell(\theta,y_{-g})\}\,\dd\theta$
give the exact identity
\begin{equation}
\begin{aligned}
-w\E_q[\ell_g]
+\KL\!\left(q\middle\|p\right)
={}&
\KL\!\left(p\middle\| \psi_{w}(\ell(\theta, y_{-g}), p_0)\right)\\
&-\log\frac{Z_{-g}(w)}{Z_y(w)}.
\end{aligned}
\label{eq:generalized-vb-identity}
\end{equation}
Therefore, by restriction over a family $\tilde{\mathcal{Q}}$
\begin{align}
\widetilde q_{g}
&=
\arg\min_{q\in \tilde{\mathcal{Q}}}
\KL\!\left(q\middle\|\psi_{w}(\ell(\theta, y_{-g}), p_0)\right)\\
&=
\arg\min_{q\in \tilde{\mathcal{Q}} }
\left\{
-w\E_q[\ell_g]
+
\KL\!\left(q\middle\|p\right)
\right\}\nonumber
\label{eq:generalized-variational-unlearning}
\end{align}

Thus ordinary and generalized variational unlearning are forward-KL 
projections of the corresponding exact deleted beliefs.

\subsection{Maxitive  Donkser-Varadhan formulation of possibilistic unlearning}
\label{sec:maxitive-unlearning}

Note that the exact possibilistic deletion rule 
\begin{equation}
\rho_{g}^* (\theta)
=\frac{\rho(\theta)e^{w\ell_g(\theta)}}
{\sup_{\theta\in\Theta}\rho(\theta)e^{w\ell_g(\theta)}}.\nonumber
\end{equation}
has the same normalized exponential form considered in the
maxitive Donsker-Varadhan formulation of \citet{singh2025maxitive}, with the 
loss $-w\ell$ replaced by its deletion counterpart  $w\ell_g$.
Their result therefore gives two variational characterizations of the
exact post-deletion belief.

For a candidate possibility function $f$, define
\begin{align}
\underline{\mathrm{CBO}}_{g}(f)
&=
\inf_{\theta\in\Theta}
\left\{
w\ell_g(\theta)
-
\log\frac{f(\theta)}{\rho(\theta)}
\right\},
\label{eq:lower-del-cbo}
\\
\overline{\mathrm{CBO}}_{g}(f)
&=
\sup_{\theta\in\Theta}
\left\{
w\ell_g(\theta)
-
\log\frac{f(\theta)}{\rho(\theta)}
\right\}.
\label{eq:upper-del-cbo}
\end{align}
Over the unrestricted space of possibility functions,
$\rho_{g}^* (\theta)$ is the unique common solution of the lower and upper
problems. Restricting the two optimization problems to $\widetilde{\mathcal{F}}$,
a tractable family of possibility functions,  gives the
maxitive variational unlearning approximations
\begin{align}
\widetilde\rho^{\,L}_{g}
&\in
\arg\max_{f\in\widetilde{\mathcal{F}}}
\underline{\mathrm{CBO}}_{g}(f)
=
\arg\min_{f\in\widetilde{\mathcal{F}}}
D_{\max}(f\|\rho_{g}^* ),
\label{eq:poss-var-lower}
\\
\widetilde\rho^{\,U}_{g}
&\in
\arg\min_{f\in\widetilde{\mathcal{F}}}
\overline{\mathrm{CBO}}_{g}(f)
=
\arg\min_{f\in\widetilde{\mathcal{F}}}
D_{\max}(\rho_{g}^* \|f).
\label{eq:poss-var-upper}
\end{align}

Equivalently, the lower and upper formulations define two 
$D_{\max}$-projections of the exact deleted possibility function onto the
tractable family $\widetilde{\mathcal F}$.
Additional details on  the max-relative-entropy identities and their
information processing interpretation are given in the Supplementary
Material.

\section{Discussion}

The paper  gives three complementary information deletion statements,
with the main result being a learning-unlearning duality, for 
probabilistic and possibistic inferential paradigms. 

In ordinary Bayesian learning, the information contribution  is 
multiplied into the prior and normalized to obtain the posterior. 
In information deletion, the same contribution is divided out from the
full posterior and the result is renormalized. 
In both directions, Bayes' theorem provides the
optimal information transformation, with two implications. 
First, every leave-data-out posterior is an exact solution 
of an information-deletion problem. This supplies a variational
characterization of leave-data-out Bayesian inference
that does not require refitting to be formulated. 
Second, approximations to data deletion can be analyzed 
through the restricted optimization problem; with
the resulting approximation error   measured 
by the KL divergence to the exact leave-data-out posterior.

The generalized bayesian unlearning result strengthens this duality.
Rather than postulating a generalized deletion objective,
we impose sequential coherence and inverse consistency
on the removal operation.
Within the same broad class of differentiable $g$-divergences 
used in coherent generalized updating, coherence forces the 
discrepancy from the current posterior to be KL divergence, while 
consistency forces the learning and deletion rate to be equal.
The deleted information enters through the negative expected loss, 
and the resulting optimizer is exactly the generalized Bayesian
posterior based on the retained data.

In possibility theory, a separate set of deletion axioms characterizes
the inverse exponential tilt with supremum normalization. These are not 
obtained by simply writing down a reciprocal likelihood rule;
the corresponding operations are derived from the structural principles 
of each calculus. In particular, Possibilistic learning has recently 
regained attention as a flexible alternative to probabilistic inference.
Our deletion framework contributes to this development by providing a
principled notion of possibilistic unlearning and opening the door
to approximate leave-data-out methods, including cross-validation for 
possibilistic models.   

Finally, these results give a unified principle: once an inferential paradigm 
specifies how information is accumulated and normalized, 
coherent deletion is the inverse transformation characterized within 
that same paradigm. Restricting the exact targets to tractable families
then produces probabilistic or maxitive variational unlearning without
changing the underlying deletion principle.



\bibliographystyle{apalike}
\bibliography{references}

@article{zellner1988optimal,
  title={Optimal information processing and Bayes's theorem},
  author={Zellner, Arnold},
  journal={The American Statistician},
  volume={42},
  number={4},
  pages={278--280},
  year={1988},
  publisher={Taylor \& Francis}
}

@article{knoblauch2022optimization,
  title={An optimization-centric view on Bayes' rule: Reviewing and generalizing variational inference},
  author={Knoblauch, Jeremias and Jewson, Jack and Damoulas, Theodoros},
  journal={Journal of Machine Learning Research},
  volume={23},
  number={132},
  pages={1--109},
  year={2022}
}

@article{nguyen2020variational,
  title={Variational bayesian unlearning},
  author={Nguyen, Quoc Phong and Low, Bryan Kian Hsiang and Jaillet, Patrick},
  journal={Advances in Neural Information Processing Systems},
  volume={33},
  pages={16025--16036},
  year={2020}
}

@article{rawat2022challenges,
  title={Challenges and pitfalls of Bayesian unlearning},
  author={Rawat, Ambrish and Requeima, James and Bruinsma, Wessel and Turner, Richard},
  journal={arXiv preprint arXiv:2207.03227},
  year={2022}
}

@article{vehtari2012survey,
  title={A survey of Bayesian predictive methods for model assessment, selection and comparison},
  author={Vehtari, Aki and Ojanen, Janne},
journal ={Statistics surveys},
  year={2012}
}

@article{blei2017variational,
  title={Variational inference: A review for statisticians},
  author={Blei, David M and Kucukelbir, Alp and McAuliffe, Jon D},
  journal={Journal of the American statistical Association},
  volume={112},
  number={518},
  pages={859--877},
  year={2017},
  publisher={Taylor \& Francis}
}

@article{bissiri2016general,
  title={A general framework for updating belief distributions},
  author={Bissiri, Pier Giovanni and Holmes, Chris C and Walker, Stephen G},
  journal={Journal of the Royal Statistical Society Series B: Statistical Methodology},
  volume={78},
  number={5},
  pages={1103--1130},
  year={2016},
  publisher={Oxford University Press}
}

@inproceedings{bourtoule2021machine,
  title={Machine unlearning},
  author={Bourtoule, Lucas and Chandrasekaran, Varun and Choquette-Choo, Christopher A and Jia, Hengrui and Travers, Adelin and Zhang, Baiwu and Lie, David and Papernot, Nicolas},
  booktitle={2021 IEEE Symposium on Security and Privacy (SP)},
  pages={141--159},
  year={2021},
  organization={IEEE}
}

@inproceedings{houssineau2026bayes,
  title={From Bayes' Rule to Bayes Rules: Information Processing, Bayes' Theorem, and Imprecise Probabilities},
  author={Houssineau, Jeremie and Ch{\'e}rief-Abdellatif, Badr-Eddine},
  booktitle={Forty-Second Annual Conference on Uncertainty in Artificial Intelligence},
  year={2026}
}

@article{singh2025maxitive,
  title={Maxitive Donsker-Varadhan formulation for possibilistic variational inference},
  author={Singh, Jasraj and Wongso, Shelvia and Houssineau, Jeremie and Ch{\'e}rief-Abdellatif, Badr-Eddine},
  journal={arXiv preprint arXiv:2511.21223},
  year={2025}
}

@article{martin2026possibilistic,
  title={Possibilistic inferential models: a review},
  author={Martin, Ryan},
  journal={Journal of the American Statistical Association},
  volume={121},
  number={553},
  pages={807--826},
  year={2026},
  publisher={Taylor \& Francis}
}

@article{zadeh1999fuzzy,
  title={Fuzzy sets as a basis for a theory of possibility},
  author={Zadeh, Lotfi A},
  journal={Fuzzy sets and systems},
  volume={100},
  pages={9--34},
  year={1999},
  publisher={Elsevier}
}

\clearpage
\appendix
\thispagestyle{empty}

\onecolumn
\aistatstitle{Supplementary Materials}

\section{Missing Proofs}

\subsection{\textbf{Proof of Proposition} \ref{prop:ordinary-kl}}
We prove that the information deletion loss is a KL divergence.

\begin{proof}
From Bayes' theorem and the decomposition $y=(y_g,y_{-g})$,
\[
\pi(y|I) = \frac{\pi(\theta|I) \pi(y|\theta)}{\pi(\theta|y)}, \text{ and } 
\pi(y_{-g}|I) = \frac{\pi(\theta|I) \pi(y_{-g}|\theta)}{\pi(\theta|y_{-g})}.
\]
Then
\[
\frac{\pi(y_{-g}|I)}{\pi(y|I)} = \frac{\pi(y_{-g}|\theta)\pi(\theta|y)} 
{\pi(\theta|y_{-g}) \pi(y|\theta)} = \frac{\pi(\theta|y)}{\pi(y_g|\theta,
y_{-g}) \pi(\theta|y_{-g})}
\]
and,
\[
\pi(\theta|y_{-g}) = \frac{\pi(\theta|y)}{\pi(y_g|\theta, y_{-g})} \Big / 
\frac{\pi(y_{-g}|I)}{\pi(y|I)} .
\]
By integrating both sides, we obtain 
\[
\frac{\pi(y_{-g}|I)}{\pi(y|I)} = \int_{\Theta}  \frac{\pi(\theta|y)}
{\pi(y_g|\theta, y_{-g})} \dd \theta = C_g^B.
\]

Finally,  replacing $\displaystyle \frac{\pi(y_{-g}|I)}{\pi(y|I)} $ into
$\Delta_g(q)$  gives:

\begin{align*} 
  \Delta_g(q)  &= \mathbb{E}_{q}\bigg[\log 
  \frac{q(\theta|y_{-g}) }{ \pi(\theta|y)}\frac{\pi(y_{-g}|I)
  \pi(y_g|\theta, y_{-g})}{\pi(y|I) } \bigg] \\
   &=\mathbb{E}_{q}\bigg[\log \frac{q(\theta|y_{-g}) }
   { \pi(\theta|y)/ \pi(y_g|\theta, y_{-g})}\frac{\pi(y_{-g}|I)}{\pi(y|I) } \bigg] \\
   &= \mathbb{E}_{q}\bigg[\log \frac{q(\theta|y_{-g}) }{
   r_g(\theta)}C_g^B \bigg] \\
\Delta_g(q) &= \KL\!\left(q(\theta|y_{-g} )\,\middle\|\,\frac{r_g(\theta)}{C_g^B}\right).
\end{align*}
\end{proof}

\subsection{Proof of Theorem \ref{thm:coherent-deletion-kl}}

We prove that the coherent deletion requirement also implies a KL divergence.

\begin{proof}
We adapt the KL characterization argument for coherent generalized updating. 
Similar to \citet{bissiri2016general}, we consider a two-point parameter 
space $\Theta=\{\theta_1,\theta_2\}$. Write
\[
P=p_0\delta_{\theta_1}+(1-p_0)\delta_{\theta_2},
\qquad
Q=p\delta_{\theta_1}+(1-p)\delta_{\theta_2},
\]
with $p_0,p\in(0,1)$. For a loss $\ell$, the deletion objective 
in \eqref{eq:generic-deletion-loss} is
\begin{align*}
\mathcal L_{-}(p;p_0,\ell)
={}&-p\,\ell(\theta_1)-(1-p)\ell(\theta_2) \\
&+p_0g\!\left(\frac{p}{p_0}\right)
 +(1-p_0)g\!\left(\frac{1-p}{1-p_0}\right).
\end{align*}
At an interior minimizer the first-order condition is
\begin{equation}
g'\!\left(\frac{p}{p_0}\right)
-g'\!\left(\frac{1-p}{1-p_0}\right)
=
\ell(\theta_1)-\ell(\theta_2).
\label{eq:supp-foc}
\end{equation}
Define
\[
A(a,b)
=
g'\!\left(\frac{a}{b}\right)
-g'\!\left(\frac{1-a}{1-b}\right).
\]
Suppose deletion of $\ell_1$ maps $p_0$ to $p_1$, and 
subsequent deletion of $\ell_2$ maps $p_1$ to $p_2$.
Equation~\eqref{eq:supp-foc} gives
\[
A(p_1,p_0)=\Delta \ell_1,
\qquad
A(p_2,p_1)=\Delta \ell_2,
\]
where $\Delta \ell_j=\ell_j(\theta_1)-\ell_j(\theta_2)$.
Sequential coherence requires $p_2$ also to be the minimizer obtained by 
deleting $\ell_1+\ell_2$ directly from $p_0$, and hence
\[
A(p_2,p_0)=\Delta \ell_1+\Delta \ell_2.
\]
Therefore
\begin{equation}
A(p_2,p_0)=A(p_1,p_0)+A(p_2,p_1).
\label{eq:supp-functionalA}
\end{equation}
Set $p_0=t$, $p_1=tx$, and $p_2=txy$. 
Substitution into \eqref{eq:supp-functionalA} yields
\begin{align*}
g'(xy)-g'\!\left(\frac{1-txy}{1-t}\right)
={}&g'(x)-g'\!\left(\frac{1-tx}{1-t}\right) \\
&+g'(y)-g'\!\left(\frac{1-txy}{1-tx}\right).
\end{align*}
Letting $t\downarrow0$ and using continuity of $g'$ gives
\[
g'(xy)-g'(1)
=
\{g'(x)-g'(1)\}+\{g'(y)-g'(1)\}.
\]
Thus $\phi(x)=g'(x)-g'(1)$ satisfies the multiplicative Cauchy equation
\[
\phi(xy)=\phi(x)+\phi(y).
\]
Continuity implies $\phi(x)=c\log x$, so
$
g'(x)=c\log x+g'(1).
$
Integrating and using $g(1)=0$ gives
$
g(x)=cx\log x+\{g'(1)-c\}(x-1).
$
The linear term integrates to zero in a $g$-divergence because
$
\int\left(\frac{\dd Q}{\dd P}-1\right)\dd P=0.
$
Consequently,
\[
D_g(Q,P)
=c\int \log\!\left(\frac{\dd Q}{\dd P}\right)\dd Q
=c\,\KL(Q\|P).
\]
Convexity gives $c\ge0$, and uniqueness of the minimizer rules out 
the degenerate case $c=0$. Hence $c>0$,
proving Theorem~\ref{thm:coherent-deletion-kl}.
\end{proof}

\subsection{Proof of Corollary \ref{cor:eq-leave-bayes}}
We prove the equivalence between the Bayes deletion rule and the leave-data-out
posterior.

\begin{proof}
From Bayes' theorem 
$\displaystyle \pi(\theta|y) = \frac{\pi(\theta|y_{-g})\ell_g(\theta)}
{ \pi(y_{g}|y_{-g})}$

Also,
$ \displaystyle \int_\Theta \frac{\pi(\theta|y)}{ \ell_g(\theta)}\pi(y_{g}|y_{-g}) 
\dd\theta = 1$ and 
$$
\pi(y_{g}|y_{-g}) = 1\bigg/ \int_\Theta \frac{\pi(\theta|y)}{ \ell_g(\theta)}
\dd\theta.
$$

and rearranging terms gives
$$
\pi(\theta\mid y_{-g})
=\frac{\pi(\theta\mid y)/\ell_g(\theta)}
{\int \pi(\theta\mid y)/\ell_g(\theta)\dd\theta} = q^*(\theta|y_{-g}).
$$
\end{proof}

\subsection{Proof of Corollary \ref{cor:generalized-unlearning}}
We prove the equivalence between the generalized Bayesian deletion rule 
and the generalized Bayesian updating based on the retained loss.

\begin{proof}
Starting from the prior $p_0$,
the generalized posterior based on the loss $\ell_{-g}(\theta)= 
\ell(\theta, y_{-g})$ is given by
$$
\psi(w\ell_{-g}(\theta), p_0) = \frac{
p_0(\theta)\exp\{-w\ell_{-g}(\theta)\}
}{
\int p_0(\theta)\exp\{-w\ell_{-g}(\theta)\}\,\dd\theta
}.
$$
Since, the loss is  additive, i.e. $\ell(\theta,y) = \ell_g(\theta)+
\ell_{-g}(\theta)$, and by replacing $p(\theta)$ in 
\eqref{eq:optimal-gen-deletion}
\begin{align*}
 q_g^*(\theta)&=\frac{p(\theta)\exp\{w\ell_g(\theta)\}
}{\int_{\Theta}p(\theta)\exp\{w\ell_g(\theta)\}\dd\theta} \\
&= \frac{p_0(\theta)\exp\{-w\ell_{-g}(\theta)\}}{
\int p_0(\theta)\exp\{-w\ell_{-g}(\theta)\}\,\dd\theta }\\
&= \psi_{w}(\ell_{-g}(\theta), p_0).
\end{align*}

\end{proof}

\subsection{Proof of Corollary \ref{cor:possibilistic-unlearning}}
We prove the equivalence between the possibilistic deletion rule and the
possibilistic updating rule based on the retained information.

\begin{proof}
Partition the total loss as $\ell(\theta, y)= \ell_{-g}(\theta)+\ell_g(\theta)$.
Replacing 
$\displaystyle 
\rho(\theta) = \frac{f_0(\theta)e^{-w\ell(\theta,y)}}{Z_{\max}(w;f_0,\ell)}$
in the  deletion rule gives
\begin{align*}
\rho_{g}^* (\theta)
&=
\frac{\rho(\theta)e^{w\ell_g(\theta)}}
{\sup_{\theta\in\Theta}\rho(\theta)e^{w\ell_g(\theta)}} \\
&= \frac{
f_0(\theta)e^{-w\ell_{-g}(\theta)}
}{
\sup_{\theta\in\Theta}
f_0(\theta)e^{-w\ell_{-g}(\theta)}
}.
\end{align*}

\end{proof}

\section{Derivation of maxitive variational possibilistic unlearning}
\label{app:maxitive-unlearning}

We derive the variational representation used in
Section~\ref{sec:maxitive-unlearning}. Let
\[
\rho_{g}(\theta)
=
\frac{
\rho(\theta)\exp\{w\ell_g(\theta)\}
}{
C_g^{\max}(w;\rho)
},
\]
where
\[
C_g^{\max}(w;\rho)
=
\sup_{\theta\in\Theta}
\rho(\theta)\exp\{w\ell_g(\theta)\},
\]
and assume
$0<C_g^{\max}(w;\rho)<\infty$.

\subsection{Connection with the maxitive Donsker-Varadhan representation}

The maxitive Donsker-Varadhan result of \citet{singh2025maxitive} applies to
normalized exponential tilts of the generic form
\[
\frac{f(\theta)e^{h(\theta)}}
{\sup_{\theta}f(\theta)e^{h(\theta)}}, \; f(\theta) \in \mathcal{F}(\Theta).
\]
The usual possibilistic learning case corresponds to $h(\theta)=-\ell(\theta)$.
For deletion we take
\[
f(\theta)=\rho(\theta),
\qquad
h=w\ell_g.
\]
Hence the corresponding pointwise variational quantity is
\begin{equation}
A_f(\theta)
=
w\ell_g(\theta)
-
\log\frac{f(\theta)}{\rho(\theta)}.
\label{eq:deletion-potential-functional}
\end{equation}
Its infimum and supremum give the lower and upper deletion
consistency bounds,
\[
\underline{\mathrm{CBO}}_{g}(f)
=
\inf_{\theta} A_f(\theta),
\qquad
\overline{\mathrm{CBO}}_{g}(f)
=
\sup_{\theta} A_f(\theta).
\]

\subsection{Max-relative-entropy identities}

For normalized possibility functions $f$ and $h$, define
\begin{equation}
D_{\max}(f\|h)
=
\sup_{\theta\in\Theta}
\log\frac{f(\theta)}{h(\theta)}.
\label{eq:dmax-supp}
\end{equation}

From the exact deletion rule,
\[
\log \rho_{g}^* (\theta)
=
\log \rho(\theta)
+
w\ell_g(\theta)
-
\log C_g^{\max}(w;\rho).
\]
Therefore
\begin{align}
D_{\max}(f\|\rho_{g}^* )
&=
\sup_{\theta}
\left\{
\log\frac{f(\theta)}{\rho(\theta)}
-
w\ell_g(\theta)
+
\log C_g^{\max}
\right\}
\nonumber\\
&=
\log C_g^{\max}
-
\underline{\mathrm{CBO}}_{g}(f),
\end{align}
and hence
\begin{equation}
\log C_g^{\max}(w;\rho)
=
\underline{\mathrm{CBO}}_{g}(f)
+
D_{\max}(f\|\rho_{g}).
\label{eq:maxdv-lower-del-supp}
\end{equation}

Similarly,
\begin{align}
D_{\max}(\rho_{g}^* \|f)
&=
\sup_{\theta}
\left\{
w\ell_g(\theta)
-
\log\frac{f(\theta)}{\rho(\theta)}
-
\log C_g^{\max}
\right\}
\nonumber\\
&=
\overline{\mathrm{CBO}}_{g}(f)
-
\log C_g^{\max},
\end{align}
so that
\begin{equation}
\log C_g^{\max}(w;\rho)
=
\overline{\mathrm{CBO}}_{g}(f)
-
D_{\max}(\rho_{g}\|f).
\label{eq:maxdv-upper-del-supp}
\end{equation}

Equations~\eqref{eq:maxdv-lower-del-supp} and
\eqref{eq:maxdv-upper-del-supp} show that
\[
\sup_{f\in\mathcal F(\Theta)}
\underline{\mathrm{CBO}}_{g}(f)
=
\log C_g^{\max}
=
\inf_{f\in\mathcal F(\Theta)}
\overline{\mathrm{CBO}}_{g}(f).
\]

Moreover,
\[
D_{\max}(f\|\rho_{g}^* )=0
\quad\Longleftrightarrow\quad
f\leq \rho_{g}^* 
\]
pointwise for normalized possibility functions, whereas
\[
D_{\max}(\rho_{g}^* \|f)=0
\quad\Longleftrightarrow\quad
\rho_{g}^* \leq f.
\]
Thus the lower and upper problems may each have multiple optimizers,
but their optimizer sets intersect uniquely at
$f=\rho_{g}^* $.

\subsection{Information processing interpretation}

The connection with \citet[Section~4.3]{houssineau2026bayes} can be seen
directly from the pointwise identity
\begin{equation}
\log\frac{f(\theta)}{\rho_g(\theta)}
=
\log C_g^{\max}(w;\rho)
-
A_f(\theta).
\label{eq:deletion-information-identity}
\end{equation}
In the forward learning setting, \citet{houssineau2026bayes} interpret the
analogous identity as relating the consistency bound integrand of
\citet{singh2025maxitive} to the pointwise information imbalance of the
candidate belief. Equation~\eqref{eq:deletion-information-identity}
shows that the same relation holds for information deletion after
replacing the forward potential $-\ell$ by the deletion potential
$w\ell_g$.

The lower and upper deletion CBOs therefore encode complementary
one-sided deviations from the exact deleted belief. Their unique
common optimum corresponds to the exact axiomatic deletion rule,
while restricting either problem to a tractable family
$\widetilde{\mathcal{F}} \subset\mathcal F(\Theta)$ yields
\begin{align}
\widetilde\rho^{\,L}_{g}
&\in
\arg\max_{f\in\widetilde{\mathcal{F}}}
\underline{\mathrm{CBO}}_{g}(f)
=
\arg\min_{f\in\widetilde{\mathcal{F}}}
D_{\max}(f\|\rho_{g^* }),
\\
\widetilde\rho^{\,U}_{g}
&\in
\arg\min_{f\in\widetilde{\mathcal{F}}}
\overline{\mathrm{CBO}}_{g}(f)
=
\arg\min_{f\in \widetilde{\mathcal{F}}}
D_{\max}(\rho_{g}^* \|f).
\end{align}
These are the lower and upper maxitive variational unlearning
problems used in the main text.

\end{document}